\documentclass[11pt]{article}
\usepackage{ifxetex,ifluatex}
\usepackage{CJKutf8}
\newif\ifxetexorluatex
\ifxetex
  \xetexorluatextrue
\else
  \ifluatex
    \xetexorluatextrue
  \else
    \xetexorluatexfalse
  \fi
\fi

\ifxetexorluatex
  \usepackage{fontspec}
\else
  \usepackage[T1]{fontenc}
  \usepackage[utf8]{inputenc}
  \usepackage{lmodern}
\fi

\usepackage{hyperref}

\usepackage{ifxetex,ifluatex}
\newif\ifxetexorluatex
\ifxetex
  \xetexorluatextrue
\else
  \ifluatex
    \xetexorluatextrue
  \else
    \xetexorluatexfalse
  \fi
\fi

\ifxetexorluatex
  \usepackage{fontspec}
\else
  \usepackage[T1]{fontenc}
  \usepackage[utf8]{inputenc}
  \usepackage{lmodern}
\fi
\usepackage{amsthm}
\usepackage{paithan}

\usepackage{xcolor}
\usepackage{mathtools}
\usepackage{tikz}
\usepackage{enumitem}
\usepackage[ruled,lined]{algorithm2e}
\usetikzlibrary{arrows.meta}
\usetikzlibrary{shapes, fit, backgrounds, positioning}

\usepackage{braket}

\theoremstyle{plain}
\newtheorem{definition}{Definition}[section]
\newtheorem{theorem}{Theorem}[section]

\newtheorem{corollary}{Corollary}[section]
\newtheorem{lemma}{Lemma}[section]

\newtheorem{proposition}{Proposition}[section]

\usepackage{geometry}
\usepackage{setspace}
\def\nimber{\mbox{\rm nimber}}
\def\mex{\mbox{\rm mex}}
\def\size{\mbox{\rm size}}
\newcommand{\zh}[1]{\begin{CJK}{UTF8}{gbsn}#1\end{CJK}}

\newcommand{\rsGeography}{\ruleset{Geography}}

\clearpage
\title{Nimber-Preserving Reduction: Game Secrets and\\ Homomorphic Sprague-Grundy Theorem}

\author{Kyle W. Burke
\\Plymouth State\\kwburke@plymouth.edu \and
Matthew T. Ferland
\\ USC\\mferland@usc.edu \and Shang-Hua Teng\thanks{%Department of Computer Science and Department of Mathematics, Universiry of Southern California (USC).
Supported by the Simons Investigator Award for fundamental \& curiosity-driven research and NSF grant CCF-1815254.}\\ USC\\
shanghua@usc.edu}

\begin{document}
\maketitle

\begin{abstract}
The concept of {\em nimbers}---a.k.a. {\em Grundy-values} or {\em nim-values}---is fundamental to combinatorial game theory.
%Enhancing {\em winnability} and succinctly summarizing  game trees,
Beyond the winnability, nimbers provide a complete characterization of strategic interactions among impartial games in disjunctive sums. %The celebrated Sprague-Grundy theory \cite{Sprague:1936} \cite{Grundy:1939}---foundational to combinatorial game theory---impacted many subsequent works, including Conway's on {\em numbers and games} \cite{ONAG:2000}.
In this paper, we consider  {\em nimber-preserving} reductions among impartial games, which enhance the  {\em winnability-preserving reductions} in traditional %complexity-theoretical
computational characterizations of combinatorial games.
We prove that \ruleset{Generalized Geography} is complete for the natural class, ${\cal I}^P$, of {\em polynomially-short} impartial rulesets, under polynomial-time nimber-preserving reductions.
We refer to this notion of completeness as {\em Sprague-Grundy-completeness}.
%For a family $\cal J$ of impartial rulesets, we call $R\in {\cal J}$ a {\em Sprague-Grundy}-complete ruleset for ${\cal J}$ if for any game $Z$ of a ruleset in ${\cal J}$, one can construct, {\em in polynomial time}, a game $G$ of $R$ such that the Grundy values of $Z$ and $G$ are equal.  We focus on the natural class, ${\cal I}^P$, of {\em polynomially-short} impartial rulesets, namely, those that define games with game-trees of   polynomial heights (in the games' sizes).  We prove the following structural theorem: \ruleset{Generalized~Geography}---a well-known \cclass{PSPACE}-complete ruleset---is Sprague-Grundy-complete for ${\cal I}^P$.
In contrast, we also show that not every \cclass{PSPACE}-complete ruleset in ${\cal I}^P$ is Sprague-Grundy-complete for ${\cal I}^P$.
%This theorem highlights the fundamental difference between  winnability-preserving and nimber-preserving reductions.
% It also further illuminates the central role of {\em Generalized Geography}---a classical ruleset instrumental to
% Lichtenstein-Sipser's
% \cclass{PSPACE}-hard characterization of \ruleset{Go}
% \cite{DBLP:journals/jct/FraenkelL81}---in a complexity-theoretical understanding of combinatorial games.

By viewing every impartial game as an encoding of its nimber---a succinct game secret richer than its winnability alone---our technical result establishes the following striking cryptography-inspired homomorphic theorem:
%Our polynomial-time nimber-preserving reduction
%always constructs {\em prime games}, namely, those that cannot be expressed as the disjunctive sum of two (or more) other impartial games.
%Let $\oplus$ denote the bitwise-xor operator.
%Our technical result establishes the following striking cryptography-inspired {\em homomorphic theorem}
%in combinatorial game theory:
 Despite the \cclass{PSPACE}-completeness of
nimber computation for ${\cal I}^P$,
 there exists a polynomial-time algorithm to construct,
  for any pair of games $G_1, G_2$
   in  ${\cal I}^P$,
   a \ruleset{Generalized Geography} game $G$ %---which is a natural non-trivial {\em prime game} (i.e. a game that cannot be written as a sum of two other games) of ${\cal I}^P$---
   satisfying:
$$\nimber(G) = \nimber(G_1) \oplus \nimber(G_2).$$

%Thus, if we substitute single-pile Nim with single-graph Generalized Geography in Sprague-Grundy theory, then both the reduction and summation become polynomial-time computable.

%In other words,
%the bitwise-xor of
% the \cclass{PSPACE}-hard ``nimber secrets'' encoded in any two impartial games of ${\cal I}^P$
% can be re-encoded in polynomial time
% by a prime game of  ${\cal I}^P$.
\end{abstract}

\newpage

\setcounter{page}{1}

\section{Introduction}

Mathematical games are fun and intriguing.
Even with {\em succinct} rulesets (which define game positions and the players' options from each position),
they can grow {\em game trees}
of size {\em exponential} in that
 of the starting positions.
A game is typically formulated for two players.
They take turns strategically selecting  from the current options to move the game state
to the next position.
In the {\em normal-play convention}, the player that faces a {\em terminal position}---a position with no feasible options—loses the game.
%The {\em game tree} from a starting position---with the leaves as the terminal positions---naturally captures this alternation of all potential feasible moves.
Over the years, rulesets have been formulated based on
graph theory \cite{WinningWays:2001,LichtensteinSipser:1980,ONAG:2001}, logic \cite{DBLP:journals/jcss/Schaefer78},
topology \cite{NashHex,Gale:1979,Reisch:1981,PosetGame,DBLP:journals/im/BurkeT08}, and other mathematical fields, often inspired by real-world phenomena \cite{goff2006quantum,akl2010importance,glos2019role,dorbec2017toward,wolfe2000go,burke2021transverse,applegatejacobsonsleator:1991,Zeilberger:2004}.
These rulesets distill fundamental mathematical concepts, structures, and dynamics.
For example:
\begin{itemize}
\item \ruleset{Node Kayles} \cite{DBLP:journals/jcss/Schaefer78} models a strategic game of growing a maximal independent set:
  Each position is an undirected
  graph, and each move consists of removing a node and its  neighbors.
\item \ruleset{Generalized Geography} \cite{DBLP:journals/jcss/Schaefer78,LichtensteinSipser:1980}
models a two-player game of traversing maximal paths:
 Each position is defined
  by a token in a directed graph
  and a move consists of removing the current vertex from the graph and  moving the token to an out-neighbor.  (It is often just referred to as \rsGeography.)
\item  \ruleset{Atropos} \cite{DBLP:journals/im/BurkeT08} models the dynamic formation of discrete equilibria (panchromatic triangles) in topological maps:  Positions are partially colored Sperner triangles, and a move consists of coloring a vertex.
\end{itemize}

The deep alternation of strategic reasoning also intrinsically connects optimal play in many games to
  highly intractable complexity classes, most commonly \cclass{PSPACE}.
After Even and Tarjan proved this
 for a generalization of Nash's \ruleset{Hex} \cite{EvenTarjanHex}, deciding {\em winnabilty} of many natural combinatorial games---including
\ruleset{Node-Kayles}, \ruleset{Generalized Geography},
\ruleset{Avoid True}, \ruleset{Proper-K-Coloring},
\ruleset{Atropos}, \ruleset{Graph Nim}, and \ruleset{Generalized Chomp}\footnote{We consider \ruleset{Generalized Chomp} to be \ruleset{Chomp}, but on any Directed Acyclic Graph.  This is equivalent to \ruleset{Finite Arbitrary Poset Game}, when the partial order can be evaluated in polynomial time.}---have been shown to be \cclass{PSPACE}-complete \cite{DBLP:journals/jcss/Schaefer78,LichtensteinSipser:1980,ONAG:2001,DBLP:journals/tcs/BeaulieuBD13,DBLP:journals/im/BurkeT08,PosetGame,Grier}.

\subsection{A Classical Mathematical Theory for Impartial Games}
Mathematical characterizations
of combinatorial games emerged prior to the age of modern computational complexity theory.
In 1901, Bouton \cite{Bouton:1901} developed a complete theory for
\ruleset{Nim}, based on an ancient Chinese game Jian Shi Zi (\zh{捡石子} - picking stones).
A \ruleset{Nim} position is a collection of piles of (small) stones.
On their turn, a player takes one or more stones from exactly one of the piles.
Representing each \ruleset{Nim} position by a list of integers,
Bouton \cite{Bouton:1901} proved
that  the current player has a winning strategy in the normal-play setting
if and only if the
{\em bitwise-exclusive-or} of
these integers (as {\em binary representations}) is not zero.
% Note that although the game tree of a \ruleset{Nim} position could be exponentially tall and doubly exponentially large in the number of bits representing the position, Bouton's characterization provides a polynomial-time solution for determining the winnability of \ruleset{Nim} games.

\ruleset{Nim} is an example of an
{\em impartial} ruleset,
 meaning both players have the
 same options at every position.
Games that aren't impartial are known as {\em partisan}.
The two graph games, \ruleset{Node-Kayles} and \ruleset{Generalized Geography} aforementioned, are also impartial.

In the 1930s, Sprague \cite{Sprague:1936} and Grundy \cite{Grundy:1939} independently
developed a comprehensive mathematical theory for impartial games.
They introduced a notion of {\em equivalence} among games, characterizing their contributions in the disjunctive sums with other impartial games.
Extending Bouton's theory for \ruleset{Nim}, Sprague-Grundy Theory provides a complete mathematical solution
to the disjunctive sums of impartial games.
%  ``algebraic'' system of  of combinatorial games.

\begin{definition}[Disjunctive Sum]
For any two game positions  $G$ and $H$ (respectively, of rulesets $R_1$ and $R_2$), their {\em disjunctive sum},  $G+H$, is
a game position in which %at each turn,
the current player chooses to make a move in exactly one of $G$ and $H$, leaving the other unchanged.
A sum position $G+H$ is terminal if and only if
   both $G$ and $H$ are terminal according to their %corresponding
   own rulesets.
\end{definition}

Sprague and Grundy showed that every impartial game $G$ can be
equivalently replaced by a single-pile Nim game in any disjunctive sum involving $G$.
Thus, they characterized each impartial game by a natural number---now known as the {\em nimber}, {\em Grundy value}, or {\em nim-value}---which  corresponds to a number of stones in a single-pile of Nim.
Mathematically, the nimber of $G$, which we denote by $\nimber(G)$, can be
recursively formulated via $G$'s game tree:
%\begin{itemize}
%\item
(1) if $G$ is terminal, then $\nimber(G) = 0$;
otherwise
%\item
(2) if $\{G_1,...,G_{k}\}$ is the
  set options of $G$, letting $\mex$
  returns the smallest value of
  $\left(\mathbb{Z}^+\cup\{0\}\right)\setminus \left\{
  \nimber(G_1),...,\nimber(G_{t})
  \right\}$, then:
\begin{eqnarray}
\label{Eqn:nimber}
\nimber(G) = \mex\left(\left\{ \nimber(G_1),...,\nimber(G_{k})
\right\}\right)
\end{eqnarray}
%\end{itemize}
Let $\oplus$ denote the {\bf\em bitwise xor} ({\bf\em nim-sum}).
By Bouton's \ruleset{Nim}
theory \cite{Bouton:1901}:
\begin{eqnarray}\label{SGE}
\nimber(G+H) = \nimber(G)\oplus \nimber(H) \quad\quad
\mbox{$\forall$ impartial $G$, $H$}
\end{eqnarray}
Thus, Sprague-Grundy Theory---using Bouton's \ruleset{Nim} solution---provides an instrumental mathematical summary (of the much larger game trees) that enhances the winnability for impartial games:
A position is a winning position if and only if its nimber
is non-zero.
This systematic framework  inspired subsequent work
in the field,
including Berlekamp, Conway, and Guy's \textit{Winning Ways for Your Mathematical Plays} \cite{WinningWays:2001}, and Conway's {\em On Numbers And Games} \cite{ONAG:2001}, which laid the foundation for Combinatorial Game Theory (CGT).
This 1930s theory
 also has an algorithmic
 implication.
%Because the bitwise-xor is linear-time computable,
Equation (\ref{SGE})
provides a polynomial-time framework for computing the nimber of a sum game---and the hence the winnability---from the nimbers of its component games: If the nimbers of two games $G$ and $H$ are tractable,
then $\nimber(G+H)$ is also tractable.

\subsection{Our Main Contributions}

Obviously, in spite of this algorithmic implication, Sprague-Grundy Theory does not provide a general-purpose polynomial-time solution for all impartial games, as witnessed by many \cclass{PSPACE}-hard rulesets, including
\ruleset{Node Kayles} and \ruleset{Generalized Geography} \cite{DBLP:journals/jcss/Schaefer78,LichtensteinSipser:1980}.
If one views the nimber characterization of
  an impartial game as a reduction from that game to a single pile \ruleset{Nim}, then  Schaefer {\em et al}'s complexity results  demonstrate that this reduction has intractable
  constructibility.
In fact, a recent result \cite{BurkeFerlandTengGrundy}
proved that the nimber of
polynomial-time solvable \ruleset{Undirected Geography}---i.e., \ruleset{Generalized Geography} on undirected graphs---is also \cclass{PSPACE}-complete to compute.
The sharp contrast between the complexity of winnability and nimber computation illustrates a fundamental mathematical-computational
divergence in Sprague-Grundy Theory \cite{BurkeFerlandTengGrundy}:
Nimbers can be  \cclass{PSPACE}-hard ``{\em secrets of deep alternation}''
even for polynomial-time solvable games.

Computational complexity theory often gives new perspectives of classical mathematical results.
In this work, it also provides us with a new lens for understanding this classical mathematical characterization as well as  tools for exploring and identifying new fundamental characterizations in combinatorial game theory.

\subsubsection{Polynomial-Time Nimber-Preserving Reduction to Generalized Geography}

In this paper, we consider the following natural concept of
  reduction among impartial games.

\begin{definition}[Nimber-Preserving Reduction\footnote{This natural concept of reduction in combinatorial game theory
can be viewed as the analog
of  {\em functional-preserving reductions} in various fields.
To name a few:
approximation-preserving, gap-preserving, structure-preserving reductions in complexity and algorithmic theory,
hardness-preserving and security-preservation in cryptography,
%\cite{Goldreich},
dimension-preserving, metric-preserving, and topology-preserving reductions in data analytics,
parameter-preserving reductions
in dynamic systems,
counterexample-preserving reductions in model checking,
query-number-preserving,
sample-preserving and high-order-moment-preserving in statistical analysis, and
modularity-preserving reductions in network modeling.
We are inspired by several of
these works.}]
A map $\phi$ is a nimber-preserving reduction from impartial ruleset $R_1$
to impartial ruleset $R_2$ if for every position  $G$ of $R_1$,
$\phi(G)$ is a position of $R_2$ satisfying $\nimber(G)=\nimber(\phi(G))$.
\end{definition}

Because an impartial position is a losing position if and only if its nimber is zero, nimber-preserving reductions enhance  {\em winnability-preserving reductions} in traditional complexity-theoretical characterizations of combinatorial games \cite{EvenTarjanHex,DBLP:journals/jcss/Schaefer78}:
%Every nimber-preserving reduction is a winnability-preserving reduction, because
Polynomial-time nimber-preserving reductions introduce the following natural notion of ``universal'' impartial rulesets.

\begin{definition}[Sprague-Grundy Completeness]
For a family $\cal J$ of impartial rulesets,
we say $R\in {\cal J}$ is a {\em Sprague-Grundy-complete ruleset} for ${\cal J}$ if
for any position $Z$ of ${\cal J}$, one can construct, {\em in polynomial time}, a position $G\in R$ such that $\nimber(G) = \nimber(Z)$.
Furthermore, if the polynomial-time nimber-preserving reduction always
produces a prime game, then we say $R\in {\cal J}$ is a {\bf prime} Sprague-Grundy-complete ruleset for ${\cal J}$.
\end{definition}

%Because nimbers provide more detailed summary of game data, nimber-Preserving Reduction

As the main technical contribution of this paper,
we prove the following theorem regarding the expressiveness of \ruleset{Generalized Geography}.
The natural family of rulesets containing \ruleset{Generalized Geography} is ${\cal I}^P$, the  family of all impartial rulesets whose positions have game trees with height
polynomial in the sizes of the positions.
We call games of ${\cal I}^P$ {\em polynomially-short} games.
In addition to \ruleset{Generalized Geography},
${\cal I}^P$ contains many combinatorial rulesets studied in the literature, including
\ruleset{Node Kayles}, \ruleset{Chomp},
\ruleset{Proper-K-Coloring},
\ruleset{Atropos}, and
\ruleset{Avoid True}, as well as
\ruleset{Nim} and \ruleset{Graph Nim} with polynomial numbers of stones. %, {\em etc etc}.

% We design polynomial-time {\em nimber-preserving} reductions, to support this proof.

\begin{theorem}[A Complete Geography]
\ruleset{Generalized Geography}
\label{theo:GGComplete}
is a Sprague-Grundy complete  ruleset
for ${\cal I}^P$.
\end{theorem}

In other words, for example,  given any \ruleset{Node Kayles} or \ruleset{Avoid True} game, we can, in polynomial time, generate
a \ruleset{Generalized Geography} game with the same Grundy value,
despite the fact that  the Grundy value of the
input game could be intractable to compute.

Because nimber-preserving reductions generalize winnability-preserving reductions,
every Sprague-Grundy complete ruleset for ${\cal I}^P$ must be
\cclass{PSPACE}-complete to solve.
However, for a simple mathematical reason, we have the following observation:
\begin{proposition}
\label{Prop:NotForEveryone}
Not every \cclass{PSPACE}-complete ruleset in ${\cal I}^P$ is Sprague-Grundy complete for ${\cal I}^P$.
\end{proposition}

In particular, %the variant of
\ruleset{Atropos} %TODO: There is not a variant defined there that is DEFINITELY not SG-Complete.  We need a different ruleset here...
 \cite{DBLP:journals/im/BurkeT08}
is \cclass{PSPACE}-complete
but not Sprague-Grundy-complete for ${\cal I}^P$.
Thus, together Theorem \ref{theo:GGComplete} and Proposition
\ref{Prop:NotForEveryone}
highlight the fundamental difference between
winnability-preserving reductions and nimber-preserving reductions.
Our result further illuminates the central role of {\em Generalized Geography}---a classical \cclass{PSPACE}-complete ruleset instrumental to
Lichtenstein-Sipser's
\cclass{PSPACE}-hard characterization of \ruleset{Go}
\cite{DBLP:journals/jct/FraenkelL81}---in the complexity-theoretical understanding of combinatorial games.

In Section \ref{Sec:Final}, when discussing related questions, we also demonstrate the brief corollary:

\begin{corollary}
\ruleset{Edge Geography}
(formulated in \cite{DBLP:journals/jcss/Schaefer78} and studied in \cite{DBLP:journals/tcs/FraenkelSU93},
see Section \ref{sub:SGC} for the ruleset)
is also Sprague-Grundy-complete for $\cal{I}^P$.
\end{corollary}

\subsubsection{Game Secrets: Homomorphic Sprague-Grundy Theorem}

In the framework of disjunctive sums, every impartial game $G$ encodes a secret, i.e., its Grundy value $\nimber(G)$, which succinctly summarizes $G$'s game tree and can be represented by a single-pile \ruleset{Nim}.
Once this secret is obtained, by Sprague-Grundy theory,
 one can replace $G$ by its equivalent single-pile \ruleset{Nim}
 in any disjunctive sum involving $G$.
Even though \ruleset{Nim} is expressive enough to provide natural game representations of these game secrets, it does not admit an efficient reduction, even for polynomial-time solvable games, such as \ruleset{Undirected Geography},  in $\cal{I}^P$.

In contrast, for impartial games in $\cal{I}^P$, Theorem \ref{theo:GGComplete} shows that \ruleset{Generalized Geography} provides natural game representations of these game secrets.
In conjunction with Sprague-Grundy theory,
  this \ruleset{Generalized Geography}-based encoding of
  nimbers leads to a surprising cryptography-inspired homomorphic characterization of impartial games.

\begin{theorem}[Homomorphic Sprague-Grundy-Bouton Theorem]
\label{theo:Homomorphic}
$\cal{I}^P$
%There exists a natural family ${\cal H}$ of impartial rulesets (to be given below)  satisfying the
enjoys the following two contrasting properties:
\begin{itemize}
    \item {\bf Hard-Core Nimber Secret}: The problem of computing the nimber---i.e., finding $\nimber(G)$  given a position $G$ of $\cal{I}^P$---is \cclass{PSPACE}-complete.
 \item {\bf Homomorphic Game Encoding}:
  For any pair of positions $G_1$ and $G_2$ of $\cal{I}^P$,
      one can, in polynomial-time (in the sizes
      of $G_1$ and $G_2$), construct a
      \ruleset{Generalized Geography} game $G$, such that:
    $$\nimber(G) = \nimber(G_1) \oplus \nimber(G_2).$$
\end{itemize}
\end{theorem}

Like the Sprague-Grundy theory---which represents
  the game values by natural games---Theorem \ref{theo:Homomorphic} encodes ``nimber secrets''
with natural games.
The former uses ``single-pile'' \ruleset{Nim} and the later uses ``single-graph'' \ruleset{Generalized Geography}.
For both, one can compute, in polynomial time, the representation
 of the disjunctive sum of any two representations.
Because \ruleset{Generalized Geography} is not naturally closed
  under disjunctive sums,
  genuine computational effort---although feasible in
  polynomial-time---is required in the homomorphic encoding
  of the disjunctive sum.

\newcommand{\textScale}{1.5}

\begin{figure}[h!]
\begin{center}
\scalebox{.8}{
\begin{tikzpicture}
    \node[scale=1.8] (ApB)  {$A+B$};
    %\node[] (spaceA) [right=of ApB] {};
    \node[] at (5,0)(twoGames)  {Two Equivalent Games};
    %\node[] (spaceB) [right=of twoGames] {};
    \node[] at (13, 0) (oneGame)  {One Equivalent Game};
    \node[rectangle, align=center] (twoNims) [above=of twoGames] {\phantom{x} \\ \scalebox{\textScale}{\ruleset{Nim}$(a)\ +$ \ruleset{Nim}$(b)$} \\ where \ruleset{Nim}$(a) = A$\\and \ruleset{Nim}$(b) = B$};
    \node[rectangle, align=center, scale=\textScale] (oneNim) [above=of oneGame] {\ruleset{Nim}$(a \oplus b)$};
    %\node[rectangle, align=center] (twoGeographies) [below=of twoGames] {\phantom{X}\\\phantom{X}\\\scalebox{\textScale}{\ruleset{Geog}$(G_A) + $ \ruleset{Geog}$(G_B)$} \\ where \ruleset{Geog}$(G_A) = A$\\ and \ruleset{Geog}$(G_B) = B$};
    \node[rectangle, scale=\textScale] (oneGeography) [below=of oneGame] {\ruleset{Geography}$(G_C)$};

    %dash pattern=on 1pt off 4pt on 6pt off 4pt

    \path[->]
        (ApB) edge [bend left] node [above, sloped, align=center] {Sprague \& Grundy \\(\cclass{PSPACE}-hard)} (twoNims)
        (twoNims) edge node [above, sloped, align=center] {Bouton \\ (in \cclass{P})} (oneNim)
        (ApB) edge [bend right=10] node [below, sloped, align=center] {Section \ref{Sec:SPComplete}\\(in \cclass{P})} (oneGeography)
        %(ApB) edge [bend right] node [below, sloped, align=center] {Section \ref{Sec:SPComplete}\\(in \cclass{P})} (twoGeographies)
        %(twoGeographies) edge node [below, sloped] {in \cclass{P}} (oneGeography);
    ;

\end{tikzpicture}}\end{center}
\caption{Two transformations of impartial games $A$ and $B$ into a single game equivalent to their disjunctive sum.}
\label{fig:twoTransformations}
\end{figure}
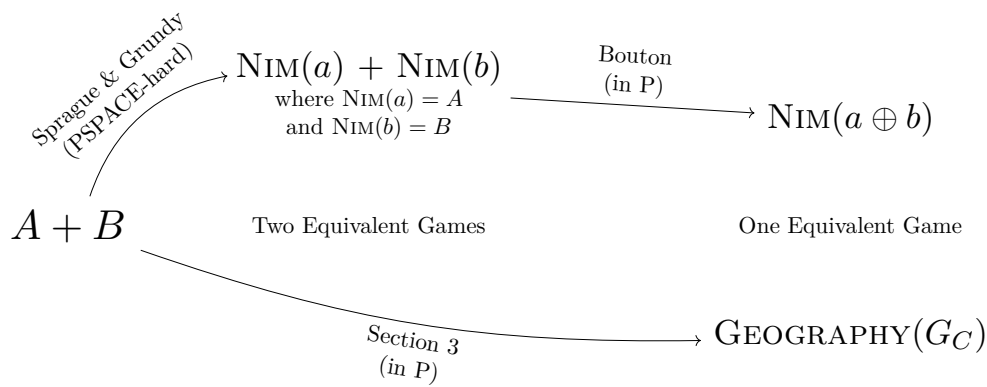

\subsection{Remarks}

This genuineness can be partially captured by Conway's notation of {\em prime impartial games} in his studies of misere games.
Let prime impartial games
be ones that can't be expressed as the disjunctive sum of any two other games (see
Section \ref{Sec:PrimeGames} for the formal definition).
Like the role of prime numbers
in the multiplicative group over integers, these games are the basic building blocks in the disjunctive-sum-based monoid over impartial games.

In Theorem \ref{theo:Prime}, we will prove that each \ruleset{Generalized Geography} game created in the proof of
Theorem \ref{theo:Homomorphic} is indeed a prime game.
Thus, in Theorem \ref{theo:Homomorphic},
  because $G$ is a natural prime game, the algorithm for ``Homomorphic Game Encoding'' cannot trivially output the syntactic description of $G_1+G_2$.
Thus, it must use a more ``elementary''
position to encode $\nimber(G_1)\oplus\nimber(G_2)$.
 %``algebraically''
In other words, the bitwise-xor of
 the \cclass{PSPACE}-hard ``nimber secrets'' encoded in any two impartial games of $\cal{I}^P$
 can be efficiently re-encoded
 by a natural prime game of $\cal{I}^P$,
 whose game tree is
 not isomorphic to that of $G_1+G_2$.

Conway's notation of prime games only partially captures the genuineness of the homomorphic encoding in Theorem \ref{theo:Homomorphic} because one may locally modify the game rule for some zero positions to make the sum game prime without changing the nimber. So the naturalness of \ruleset{Generalized Geography} captures more beyond the current concept of prime games.
We continue to look for a more accurate characterization.

%In this paper, we prove that ${\cal H} = {\cal I}^P$ satisfies Theorem  \ref{theo:Homomorphic}.

Note also that the
characterization presented
in Theorem
\ref{theo:Homomorphic} is cryptography-inspired
rather than cryptographically-applicable.
In %the crypto-notion of
 {\em partially homomorphic encryption},
the encoding functions, such as RSA encryption and discrete-log, must be one-way functions.
Here, we consider a game position as a ``natural encoding'' of its nimber-secret:
The focus of Theorem
\ref{theo:Homomorphic} is on
the complexity-based homomorphic property of this encoding  (and  that it arises naturally from impartial games) rather than the construction of
homomorphic ``one-way'' game-based encryption of secret messages.  (See Section, \ref{Sec:Final}---Conclusion and Open Questions---for more discussion on this.)
Thus, in contrast to  (partially) homomorphic cryptographic functions---such as discrete-log and RSA---whose secret messages can be recovered in \cclass{NP} (due to a one-way encoding of secrets), ``decoding'' the nimber-secrets---as they are
naturally encoded in combinatorial games of  $\cal{I}^P$---is \cclass{PSPACE}-complete in the worst case.

%It gives us cool structural information about the game!

\section{Impartial Games and Their Trees: Notation and Definitions}

In this section, we review some background concepts and notation.
In the paper, we  use $\cal I$ to denote the family of all impartial rulesets; we use $\mathbb{I}$ to denote the family of all impartial games---positions---defined by rulesets in $\cal I$;
we use $\mathbb{I}^P$ to denote the family of all polynomially-short impartial games, i.e., positions defined by rulesets in ${\cal I}^P$.

Each impartial ruleset $R \in {\cal I}$
has two mathematical components
$({\cal B}_R,\rho_R)$, where
${\cal B}_R$ represents the set of all possible game positions in $R$ and
$\rho_R: {\cal B}_R \rightarrow
2^{{\cal B}_R}$
defines the options for each position in $R$.
For each position $G \in {\cal B}_R$,
  the ruleset $R$ defines a natural {\em game tree}, $T_G$, rooted at $G$. %, capturing the alternation in the game starting at $G$.
$T_G$ recursively branches with feasible options.
The root is associated with position $G$ itself. The number of children that the root has is equal to the  number of options, i.e., $|\rho_R(G)|$, at $G$. Each child is associated with a position from $\rho_R(G)$ and its sub-game-tree is defined recursively.
Thus,  $T_G$ contains all
  {\em reachable positions} of $G$ under ruleset $R$.
The leaves of $T_G$
are terminal positions under the normal-play setting.
Up to isomorphism, the game tree for an impartial game
is unique. %TODO: is this sentence necessary?

\begin{definition}[Game isomorphism]
Two impartial games, $G$ and $H$, are {\em isomorphic} to each other if  $T_G$ is isomorphic to  $T_H$.
\end{definition}

For any two impartial games $F$ and $G$, the game tree $T_{(F+G)}$ of their disjunctive sum can be more naturally
characterized via $T_F\square T_G$,
the {\em Cartesian  product}
of $T_F$ and $T_G$.
Clearly, $T_F\square T_G$ is a directed acyclic graph (DAG) rather than a rooted-tree,
 as a node may have up to two parents.
 The game tree $T_{(F+G)}$ can be viewed as a tree-expansion of
 DAG $T_F\square T_G$.
 To turn it into a tree, one may simply duplicate all subtrees whose root has two parents, and give the parents edges to different roots of those two subtrees.
We will use $T_F \blacksquare T_G$ to denote this tree expansion of the Cartesian  product $T_F \square T_G$, and call it the {\em tree sum} of $T_F$ and $T_G$.

%\subsection{The Size and Complexity of an Impartial Game}

In combinatorial game theory (CGT),
 each ruleset usually represents an infinite family of
  games, each defined by its starting position.
For algorithmic and complexity
 analyses, a {\em size} is associated with each
game position as the basis for measuring complexity  \cite{DBLP:journals/jcss/Schaefer78,fraenkel2004complexity,PapadimitriouBook:1994,BurkeFerlandTengQCGT}.
Examples include:
%\begin{itemize}
%\item
(1) the number of vertices in the graph for \ruleset{Node Kayles}
 and \ruleset{Generalizad Geography},
%\item
(2) the board length of
 \ruleset{Hex} and \ruleset{Atropos},
 and
%\item
(3) the number of bits encoding \ruleset{Nim}.
%\end{itemize}

The size measure is assumed to be {\em natural} \footnote{In other words, the naturalness assumption
rules out rulesets with embedded hard-to-compute predicate like---as a slightly dramatized illustration---``If \cclass{P} $\neq$ \cclass{PSPACE} is true, then the feasible options of a position include a special position.''} with respect to the key components of the ruleset.
In particular, for each position
$G$ in a ruleset $R$:
\begin{itemize}
\item
$G$ has a binary-string representation of length polynomial in $\size(G)$.
\item Each position reachable from $G$ has size upper-bounded by a polynomial function in $\size(G)$.

%\item At any position $F\in {\cal B}_R$, the number of options, $|\rho_R(F)|$ is bounded from above by the exponential of a polynomial function of $\size(F)$.

\item Determining if a position of $F \in {\cal B}_R$  is an option of $G$---i.e., whether $F \in  \rho_{R}(G)$---takes time polynomial in $size(G)$.
\end{itemize}

Recall that the family, ${\cal I}^P$, discussed in the introduction is formulated based on the sizes of game positions.

\begin{definition}[Polynomially-Short Games]
A combinatorial ruleset $R = ({\cal B}_R, \rho_R)$ is {\em polynomially short} if the height of the game tree $T_G$ of each position $G \in {\cal B}_R$ is polynomial in  $\size(G)$.
Furthermore, we say $R\in {\cal I}^P$ is {\em polynomially-wide}
if for each position $G\in {\cal B}_R$, the number of options
$|\rho_R(G)|$ is bounded by a polynomial function in $\size(G)$.
\end{definition}
We call games of ${\cal I}^P$ {\em polynomially-short} games.
\ruleset{Generalized Geography} and
\ruleset{Node Kayles} are among the many examples of games that are both polynomially-wide and polynomially-short.
\ruleset{Nim}, however, is neither
polynomially-wide nor polynomially-short due to the binary encoding of the piles.
In general, under the aforementioned assumption regarding the size function of game positions, $|\rho_R(G)|$ could be exponential in $\size(G)$.
However, positions in $\rho_R(G)$ can be enumerated in  polynomial space.
Therefore, by DFS evaluation of game trees and classical
complexity analyses of
\ruleset{Node Kayles}, \ruleset{Generalized Geography} and \ruleset{Avoid True}:

\begin{proposition}[\cclass{PSPACE}-Completeness]
\label{PCNimber}
For any polynomially-short impartial ruleset $R$ and a position $G$ in $R$,
$\nimber(G)$ can be computed in
space polynomial in $\size(G)$.
Furthermore, under Cook-Turing reductions, nimber computation for \textbf{some} games in
${\cal I}^P$ is \cclass{PSPACE}-hard.
\end{proposition}

An impartial ruleset is said to be {\em polynomial-time solvable}---or simply, {\em tractable}---if there is a polynomial-time algorithm to identify a winning option
whenever there exists one (i.e., for the search problem associated with the decision of winnability). If one doesn't exist, then the algorithm needs to only identify this.

%\subsection{Construction of Impartial Games}

%For an impartial ruleset $R$, its game transition function $\rho_R$ can be viewed as a family of transition functions:
%$(\rho_{R}^{(n)}) | n \in \mathbb{Z}^+$
%parameterized by the sizes.

%TODO: Kyle needs to review the rest of this section, starting from here.

%In our technical results presented in this paper---both for polynomial-time Sprague-Grundy completeness and homomorphic game encoding---our algorithms construct
%positions.
%Once we construct a graph $G = (V,E)$ and the location of token $s\in V$, the options,
%$\rho_{\ruleset{Generalized Geography}}((G,s))$

%when we say a polynomial-time algorithm
%$A$ constructs an impartial game $G$ when given
%an input $Z$ of size $n$, we mean $A(Z)$
%yields a polynomial-time function (e.g., a Turing machine or polynomial-sized circuit)
%that can

%\subsection{Notation}

%\section{Geography is hard for the class}
\section{Star Atlas: A Complete Generalized Geography}
\label{Sec:SPComplete}

In our analysis, we will use the  standard CGT notation for nimbers: $*k$ for $k$, except that $*$ is shorthand for $*1$ and 0 is shorthand for $\ast 0$.\footnote{The reason for the $\ast 0 = 0$ convention is that it is equivalent to the integer zero in CGT.}
%Return to the $*k$ notation of CGT.
Mathematically, one can view $*$ as a map  from $\mathbb{Z}^+\cup \{0\}$ to (infinite) subfamilies of impartial games in $\mathbb{I}$, such that for each $k \in \mathbb{Z}^+\cup \{0\}$, $\nimber(G) = k, \forall G = *k$.
In other words, $*$ is {\em nature’s game encoding} of non-negative integers.

Sprague-Grundy Theory establishes that each impartial game's strategic
relevance in disjunctive sums is determined by
its nimber (i.e., its  star value). In this section, we prove Theorem \ref{theo:GGComplete}, showing how to use \ruleset{Generalized Geography} to efficiently ``map out'' games' nimbers across~$\mathbb{I}^P$.
%TODO: Wait, are we using Star Values instead of nimbers just for the joke in the section title?

For readability, we restate
Theorem \ref{theo:GGComplete}
to make the needed technical component explicit:

\begin{theorem}[Sprague-Grundy-Completeness of \ruleset{Generalized Geography}]
\label{Theo:Restate1.2}
There exists a polynomial-time algorithm $\phi$ such that for any game $G\in \mathbb{I}^P$,
$\phi(G)$ is a  \ruleset{Generalized Geography}
position satisfying
$\nimber(\phi(G)) = \nimber(G)$.
\end{theorem}

%We prove this by reducing from the nimber of any polynomially short game, $G$, to \ruleset{Generalized Geography}.
Our proof starts with the following basic property of
nimbers,
which follows from the recursive
definition
(given in Equation \ref{Eqn:nimber}):

\begin{proposition}
For any impartial game $G$,
 $\nimber(G)$ is bounded above by both the height of its game tree, $h$, and the number of options at $G$, $l$. In other words,
$G = \ast k \text{, where } k \leq \min(h, l).$
\end{proposition}

To simplify notation, we let $g = \min(h, l)$.  To begin the reduction, we will need a reduction for each decision problem $Q_i=$ ``Does $G = \ast i$?'' ($\forall i \in [g]$).
By Proposition \ref{PCNimber},
 the decision problem $Q_i$ is
 in \cclass{PSPACE}.
Thus, we can reduce each $Q_i$ to an instance in the \cclass{PSPACE}-complete QSAT (Quantified SAT), then to
a \rsGeography\ instance using the classic reduction, $f$, from \cite{DBLP:journals/jcss/Schaefer78,LichtensteinSipser:1980}.  Referring to the starting node of $f(Q_i)$ as $s_i$, we add two additional vertices, $a_i$ and $b_i$, with directed edges $(b_i, a_i)$ and $(a_i, s_i)$.  Now,

\begin{itemize}
    \item $s_i$ has exactly two options, so the value of $f(Q_i)$ is either $0$, $\ast$, or $\ast 2$.  By the reduction, it is $0$ exactly when $G \neq \ast i$, and in $\{\ast, \ast 2\}$ when $G = \ast i$.
    \item $a_i$ has exactly one option ($s_i$), so the value of the \rsGeography\ position starting there (instead of at $s_i$) is $0$ when $G = \ast i$ and $\ast$ otherwise.
    \item $b_i$ has exactly one option ($a_i$), so the value of the \rsGeography\ position starting there is $\ast$ when $G = \ast i$ and $0$ otherwise.
\end{itemize}

Each of these constructions from $Q_i$ is shown in Figure \ref{fig:Qs}.

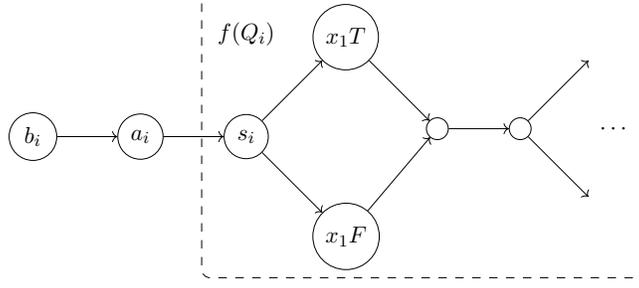
\begin{figure}[h!]
\begin{center}
\scalebox{.8}{
\begin{tikzpicture}
    \node[draw, circle] (bi) {$b_i$};
    \node[draw, circle] (ai) [right=of bi] {$a_i$};
    \node[draw, circle] (si) [right=of ai] {$s_i$};
    \node[draw, circle] (x1T) [above right=of si] {$x_1T$};
    \node[draw, circle] (x1F) [below right=of si] {$x_1F$};
    \node[draw, circle] (x12) [below right=of x1T] {};
    \node[draw, circle] (x12b) [right=of x12] {};
    \node[] (x2T) [above right=of x12b] {};
    \node[] (x2F) [below right=of x12b] {};
    \node[] (x2) [right=of x12b] {$\cdots$};
    \node[] (Q) [above=of si] {$f(Q_i)$};
    \node[fit=(Q)(x1T)(x1F)(x2F)(x2), draw, dash pattern=on 4pt off 4pt, rounded corners] {};

    %dash pattern=on 1pt off 4pt on 6pt off 4pt

    \path[->]
        (bi) edge (ai)
        (ai) edge (si)
        (si) edge (x1T)
        (si) edge (x1F)
        (x1T) edge (x12)
        (x1F) edge (x12)
        (x12) edge (x12b)
        (x12b) edge (x2T)
        (x12b) edge (x2F)
    ;

\end{tikzpicture}}\end{center}
\caption{Result of the classic QSAT and \rsGeography\ reductions of the question, $Q_i$, ``Does $G = \ast i$?'', with the added vertices $a_i$ and $b_i$.}
\label{fig:Qs}
\end{figure}

We will combine these $g+1$ \rsGeography\ instances into a single instance, but first we need some utility vertices each equal to one of the nimber values $0, \ldots, \ast (g-2)$.  We can build these using a single gadget as shown in figure \ref{fig:stars}.  This gadget consists of vertices $t_0, t_1, \ldots, t_{g-2}$ with edges $(t_i, t_j)$ for each $i > j$.  Thus, each vertex $t_i$ has options to $t_j$ where $j < i$ and no other options, exactly fulfilling the requirements for $t_i$ to have value $\ast i$.

\begin{figure}
\begin{center}
\scalebox{.8}{
\begin{tikzpicture}
    \node[draw, circle] (s0) {$t_0$};
    \node[draw, circle] (s1) [right=of s0] {$t_1$};
    \node[draw, circle] (s2) [right=of s1] {$t_2$};
    \node[draw, circle] (s3) [right=of s2] {$t_3$};
    \node[] (dots) [right=of s3] {$\cdots$};
    \node[draw, circle] (sn) [right=of dots] {$t_{g-2}$};

    \path[->]
      (s1) edge (s0)

      (s2) edge [bend left] (s0)
      (s2) edge (s1)

      (s3) edge (s2)
      (s3) edge [bend left] (s1)
      (s3) edge [bend left] (s0)

      (dots) edge (s3)

      (sn) edge (dots)
      (sn) edge [bend left] (s3)
      (sn) edge [bend left] (s2)
      (sn) edge [bend left] (s1)
      (sn) edge [bend left] (s0)
    ;

\end{tikzpicture}}\end{center}
\caption{vertices $t_0$ through $t_{g-2}$.  Each vertex $t_i$ has edges to $t_0, t_1, \ldots, t_{i-1}$.  Thus, the nimber value of the \ruleset{Geography} position at vertex $t_i$ is $\ast i$.}
\label{fig:stars}
\end{figure}
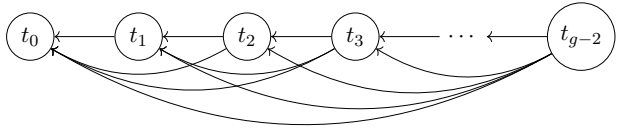

Now we build a new gadget to put it all together and combine the $f(Q_i)$ gadgets, as shown in figure \ref{fig:GeoRed}:

\begin{itemize}
    \item $\forall i \geq 1:$ add a vertex $c_i$, as well as edges $(c_i, b_i)$ and $\forall j \in [1, i-2]: (c_i, t_j)$.
    \item $\forall i \geq 2:$ add a vertex $d_i$ as well as edges $(d_i, b_1)$ and $\forall j \in [2, i-1]: (d_i, c_j)$.
    \item Finally, add a vertex $start$ with edges $(start, b_0)$, $(start, c_1)$, and $\forall j \in [2, g]: (start, d_j)$.
\end{itemize}

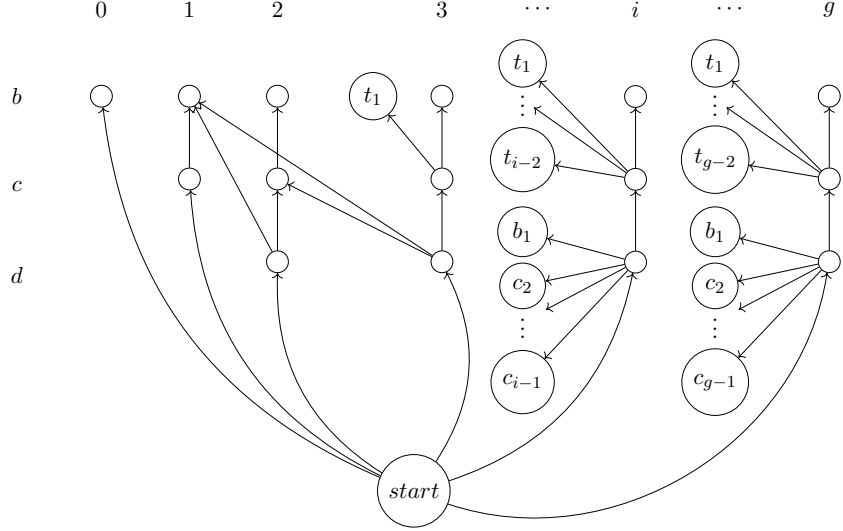
\begin{figure}
\begin{center}
\scalebox{.8}{
\begin{tikzpicture}
    \node[] (zero) {0};
    \node[draw, circle] (A0) [below=of zero] {};
    \node[] (A) [left=of A0] {$b$};
    \node[] (C) [below=of A] {$c$};
    \node[] (D) [below=of C] {$d$};

    \node[] (one) [right=of zero] {1};
    \node[draw, circle] (A1) [below=of one] {};
    \node[draw, circle] (C1) [below=of A1] {};

    \node[] (two) [right=of one] {2};
    \node[draw, circle] (A2) [below=of two] {};
    \node[draw, circle] (C2) [below=of A2] {};
    \node[draw, circle] (D2) [below=of C2] {};

    \node[] (spaceBeforeThree) [right=of two] {};
    \node[draw, circle] (t1) [right=of A2] {$t_1$};

    \node[] (three) [right=of spaceBeforeThree] {3};
    \node[draw, circle] (A3) [below=of three] {};
    \node[draw, circle] (C3) [below=of A3] {};
    \node[draw, circle] (D3) [below=of C3] {};

    \node[] (dots) [right=of three] {$\cdots$};
    \node[draw, circle] (t1b) at (7, -.9) {$t_1$};
    \node[] (tdots) at (7, -1.5) {$\vdots$};
    \node[draw, circle] (t2) at (7, -2.5) {$t_{i-2}$};
    \node[draw, circle] (a1) at (7, -3.7) {$b_1$};
    \node[draw, circle] (c2) at (7, -4.6) {$c_2$};
    \node[circle] (cdots) at (7, -5.2) {$\vdots$};
    \node[draw, circle] (c3) at (7, -6.2) {$c_{i-1}$};

    \node[] (i) [right=of dots] {$i$};
    \node[draw, circle] (Ai) [below=of i] {};
    \node[draw, circle] (Ci) [below=of Ai] {};
    \node[draw, circle] (Di) [below=of Ci] {};

    \node[] (dots2) [right=of i] {$\cdots$};
    \node[draw, circle] (t1c) at (10.2, -.9) {$t_1$};
    \node[] (tdotsb) at (10.2, -1.5) {$\vdots$};
    \node[draw, circle] (t2b) at (10.2, -2.5) {$t_{g-2}$};
    \node[draw, circle] (a1b) at (10.2, -3.7) {$b_1$};
    \node[draw, circle] (c2b) at (10.2, -4.6) {$c_2$};
    \node[circle] (cdotsb) at (10.2, -5.2) {$\vdots$};
    \node[draw, circle] (c3b) at (10.2, -6.2) {$c_{g-1}$};

    \node[] (n) [right=of dots2] {$g$};
    \node[draw, circle] (An) [below=of n] {};
    \node[draw, circle] (Cn) [below=of An] {};
    \node[draw, circle] (Dn) [below=of Cn] {};

    \node[draw, circle] (root) [below left=of c3] {$start$};

    \path[->]
      (C1) edge (A1)

      (D2) edge (C2)
      (C2) edge (A2)
      (D2) edge (A1)

      (C3) edge (A3)
      (C3) edge (t1)
      (D3) edge (C3)
      (D3) edge (C2)
      (D3) edge (A1)

      (Ci) edge (Ai)
      (Ci) edge (t1b)
      (Ci) edge (tdots)
      (Ci) edge (t2)
      (Di) edge (Ci)
      (Di) edge (a1)
      (Di) edge (c2)
      (Di) edge (cdots)
      (Di) edge (c3)

      (Cn) edge (An)
      (Cn) edge (t1c)
      (Cn) edge (tdotsb)
      (Cn) edge (t2b)
      (Dn) edge (Cn)
      (Dn) edge (a1b)
      (Dn) edge (c2b)
      (Dn) edge (cdotsb)
      (Dn) edge (c3b)

      (root) edge [bend left] (A0)
      (root) edge [bend left] (C1)
      (root) edge [bend left] (D2)
      (root) edge [bend right] (D3)
      (root) edge [bend right] (Di)
      (root) edge [bend right=50] (Dn)
    ;

\end{tikzpicture}}\end{center}
\caption{The gadget combining all of the $f(Q_i)$ gadgets into a single \rsGeography\ instance.  The vertices $b_i$, $c_i$, and $d_i$ are in rows and columns indexed by the letters on the left and the numbers along the top.  Each $b_i$ has an edge to $a_i$ as in figure \ref{fig:Qs}.}
\label{fig:GeoRed}
\end{figure}

%Suppose the height of the game tree we are reducing from is called $h$.

%First, we break the problem into $h$ subproblems, where the $i$th problem has the form of the decision problem ``Does this game have nimber $\ast i$?'' Then, because the game is polynomially short, the game is in PSPACE (HAVE LEMMA REFERENCE HERE). Thus, there is a reduction to \ruleset{QSAT}. We perform that reduction for each problem, and then perform the reduction from \ruleset{QSAT} to \ruleset{Generalized Geography}. Then, we add two vertices to each game $a$ and $b$ with edge $(a, b)$ and an edge from the $b$ to the first vertex of that game. Then, we have $h$ games that have value 0 if the QBF evaluates to false, and otherwise has value $\ast$ (since the first move has only a single option). Since exactly one of the QBFs could be true (since a game can only have a single possible nimber value), that means that all of the games of \ruleset{Geography} have value 0 except one which has value $\ast$. For the purposes of the construction, we label them with an integer $k$ based on the problem they solve of the form ``Does the game have value $\ast k$?''

\begin{lemma}
    The \rsGeography\ position starting at each vertex $c_i$ has value $\ast (i-1)$ if $G \neq \ast i$, and value 0 otherwise.
    \label{lem:cvalues}
\end{lemma}

\begin{proof}
    The \rsGeography\ position starting at $c_i$ has options to $t_j$, $\forall t \in [1, i-2]$.  That means that $c_i$ has options with values $\ast, \ldots, \ast(i-2)$.  If the move to $b_i$ has value 0, then there are moves to $0, \ast, \ldots, \ast (i-2)$, so the value at $c_i$ is $\ast (i-1)$.  Otherwise, there is no option from $c_i$ to a zero-valued position, so the value at $c_i$ is 0.
\end{proof}

\begin{lemma}
    If $G = 0$, then the \rsGeography\ position starting at each vertex $d_i$ has value $\ast i$.
    \label{lem:diG0}
\end{lemma}

\begin{proof}
    $d_i$ has moves to $b_1$, $c_2, \ldots, c_i$.  Since $G=0$, by lemma \ref{lem:cvalues} none of the vertices $c_j$ have values 0 (and $b_1$ does have value 0), so the options have values $0, \ast, \ast 2, \ldots, \ast (i-1)$, respectively. The mex of these i $\ast i$, so $d_i$ has value $\ast i$.
\end{proof}

\begin{lemma}
    Let $G = \ast k$, where $k > 0$.  Then the \rsGeography\ position starting at each vertex $d_i$ has value $\ast i$ if $i < k$, and value $\ast (k-1)$ if $i \geq k$.
    \label{lem:diGk}
\end{lemma}

\begin{proof}
    We need to prove this by cases.  We'll start with $k=1$, then show it for $k \geq 2$.

    When $k=1$, $d_i$ has moves to $b_1$, $c_2, \ldots, c_i$.  (There is no $d_0$ or $d_1$, so $i > k$.)  The value at $b_1$ is $\ast$, and by lemma \ref{lem:cvalues}, the remainder have values $\ast, \ast 2, \ldots, \ast (i-1)$, respectively.  0 is missing from this list, so $d_i = 0 = \ast (1-1) = \ast (k-1)$.

    For $k \geq 2$, we will split up our analysis into the two cases: $i < k$ and $i \geq k$.

    We will next consider the case where $k \geq 2$ and $i < k$.  From $d_i$ there are moves to $b_1$, $c_2, \ldots, c_i$.  Since $k > i$, these have values $0, \ast, \ast 2, \ldots, \ast (i-1)$, respectively, by lemma \ref{lem:cvalues}. The mex of these is $\ast i$, so $d_i$ has value $\ast i$.

    Finally, when $k \geq 2$ and $i \geq k$, $d_i$ has options to $b_1$, $c_2, \ldots, c_{k-1}, c_k, c_{k+1}, \ldots, c_i$.  By lemma \ref{lem:cvalues}, these have values $0, \ast, \ast 2, \ldots, \ast (k-2), 0, \ast k, \ldots, \ast (i-1)$, respectively.  $\ast (k-1)$ doesn't exist in that list, so that's the mex, meaning the value of $d_i$ is $\ast (k-1)$.
\end{proof}

\begin{theorem}
    Let $G = \ast k$.  Then the \rsGeography\ position beginning at $start$ equals $\ast k$.
\end{theorem}

\begin{proof}
    The options from $start$ are $b_0$, $c_1$, and $\forall i \in [2, g]: d_i$.  If $G = 0$, then $b_0$ is $\ast$, $c_1$ is $\ast$, and, by lemma \ref{lem:diG0}, each $d_i$ is $\ast i$.  Since 0 is missing from these options, the value at $start$ is $0 = \ast k$.

    If $G = \ast$, then the move to $b_0$ is 0, the move to $c_1$ is also 0, and, by lemma \ref{lem:diGk}, the moves to $d_i$ is each also 0, because each $i > k=1$ and $\ast (k-1) = \ast (1-1) = 0$.  All the options are to 0, so the value at $start$ is $\ast = \ast k$.

    Finally, if $G = \ast k$, where $k \geq 2$, then the moves are to $b_0$, $c_1$, $d_2$, \ldots, $d_{k-1}$, $d_k$, $d_{k+1}, \ldots, d_g$.  These have values, respectively, $0, \ast, \ast 2, \ldots, \ast (k-1), \ast (k-1), \ast (k-1), \ldots, \ast (k-1)$, by lemma \ref{lem:diGk}.  The mex of these is $k$, so the value of $start$ is $\ast k$.
\end{proof}

%\textbf{TODO: Move the below to section 3?}

%Now, we assume the game reduced from had a game tree height of at least 1, and thus both $b_0$ and $c_1$ are possible moves. Moving from $start$ to $b_0$ leads to a position with only a single move available, so, WLOG, $X$ has a move to a terminal position, with $Y$ having only a single move available from its starting position. Then, from $0a$, there is a move to $a_0$ which has a single move to $s_0$, but then there are 2 moves available. So the starting game tree of $Y$ has a single move, followed by another single move, followed by having two options (followed by many possibilities).

%However, the move to $c_1$ also only has a single move available. If this is another terminal move from $X$, we arrive at a contradiction, as $Y$ has 3 single moves before have two options. Thus, this needs to be a terminal move of $Y$. But again, this is a contradiction, since after moving to $b_0$, there is no terminal move for $X$.
%\end{proof}

\section{Nimber Secrets: A \cclass{PSPACE}-Complete Homomorphic Encoding}

Mathematically, Sprague-Grundy Theory together with Bouton's \ruleset{Nim} characterization
provide an algebraic view of impartial games.
Their framework establishes that the Grundy function, $\nimber()$, is a morphism from
the monoid $(\mathbb{I},+)$---impartial games with disjunctive sum---to the monoid $(\mathbb{Z}^+\cup \{0\},\oplus)$---non-negative integers in binary representations with bitwise-xor:
\begin{eqnarray*}
\nimber(G+H) = \nimber(G)\oplus \nimber(H) \quad\quad
\mbox{$\forall$  $G, H\in \mathbb{I}$}.
\end{eqnarray*}
Elegantly,
\begin{enumerate}
\item $\nimber(G)$ can also be represented by
  a natural game, i.e., a single pile \ruleset{Nim} with $\nimber(G)$ stones, and
\item the sum of two single-pile \ruleset{Nim}
  games can be represented by another single-pile \ruleset{Nim}
  whose game tree can be significantly different
  from the game tree of the sum.
\end{enumerate}
The only ``blemish''---from computational perspective---is
  that the Grundy function can be intractable to compute \cite{DBLP:journals/jcss/Schaefer78},
  even for some tractable games in $\mathbb{I}^P$ \cite{BurkeFerlandTengGrundy}.

Theorem \ref{Theo:Restate1.2} provides an alternative natural game representation of $\nimber(G)$, for $G \in \mathbb{I}^P$.
In contrast to \ruleset{Nim},
   \ruleset{Generalized Geography} admits a polynomial-time
    algorithm for computing this representation from $G$ without the need of computing $\nimber(G)$.
In fact, using \ref{Theo:Restate1.2},
  one can also compute, in polynomial time, a single-graph \ruleset{Generalized Geography} representation of
    the sum of any two (\ruleset{Generalized Geography})
    games:
\begin{theorem}[Homomorphic Game Encoding of Grundy Values]
\label{IPH}
For any pair of games $G_1, G_2 \in \mathbb{I}^P$,
  one can, in polynomial-time in $\size(G_1)+\size(G_2)$, construct a \ruleset{Generalized Geography} game $G$,
  such that:
$$\nimber(G) = \nimber(G_1) \oplus \nimber(G_2).$$
\end{theorem}

\begin{proof}
Given the game functions $\rho_{G_1}$ and $\rho_{G_2}$, in time
linear in $\size(G_1)+\size(G_2)$, one can
construct a game function $\rho_{(G_1+G_2)}$
for their disjunctive sum $G_1+G_2$ such that the game tree of $\rho_{(G_1+G_2)}$ is $G_1 \blacksquare G_2$.
This theorem follows from Theorem \ref{theo:GGComplete} and the following basic fact:
\begin{quote}
{\em The disjunctive sum $(G_1+G_2)$ of two polynomially-short games $G_1,G_2\in \mathbb{I}^P$ remains polynomially short in terms of $\size(G_1)+\size(G_2)$.
}
\end{quote}
Now we can apply Theorem \ref{theo:GGComplete} to $\rho_{(G_1+G_2)}$
to construct a prime \ruleset{Generalized Geography} $G$
in time polynomial in $\size(G_1)+\size(G_2)$.
$G$ satisfies:
$\nimber(G)= \nimber(G_1)\oplus\nimber(G_2)$.
The correctness follows from that of Theorem \ref{theo:GGComplete} and
Sprague-Grundy Theory.
\end{proof}

Figuratively,
{\em every impartial game $G$ encodes a secret, $\nimber(G)$.} The game $G$ itself can be viewed as an ``{\em encryption}'' of its nimber-secret.
The players who can uncover
this nimber-secret can play the game optimally.
For every game $G\in \mathbb{I}^P$, this secret can be ``decrypted''  by a DFS-based evaluation of $G$'s game tree in polynomial space.
Thus, computing the nimber for $\mathbb{I}^P$ is \cclass{PSPACE}-complete (under the Cook-Turing reduction).

Speaking of encryption, several basic cryptographic functions have homomorphic properties.
For example, for every RSA encryption function $\mbox{\rm ENC}_{\mbox{\footnotesize {RSA}}}$, for every pair of its messages $m_1$ and $m_2$, the following holds:
\begin{eqnarray*}
\mbox{\rm ENC}_{\mbox{\footnotesize {RSA}}}(m_1\times m_2) = \mbox{\rm ENC}_{\mbox{\footnotesize {RSA}}}(m_1)\times \mbox{\rm ENC}_{\mbox{\footnotesize {RSA}}}(m_2).
\end{eqnarray*}
Another example is the discrete-log function.
For any prime $p$, any primitive element $g\in Z^*_p$,
and any two messages $m_1, m_2\in Z^*_P$:
\begin{eqnarray*}
g^{m_1+m_2} = g^{m_1}\times g^{m_2}.
\end{eqnarray*}
Assuming RSA encryption and the discrete-log function are computationally intractable to invert, these morphisms state that without decoding the secret messages from their encoding, one can efficiently encode their product or sum, respectively, with the RSA and discrete-log functions.
In cryptography, these functions are said to support {\em partially homomorphic encryption}.

Together, Theorem \ref{IPH} and Theorem \ref{Theo:Restate1.2}
 establish that polynomially-short impartial games are themselves {\em partially homomorphic encodings} of their nimber-secrets:
Without decoding their nimbers, one can efficiently create a
\ruleset{Generalized Geography} game encoding the $\oplus$ of their nimbers.
%This following theorem instantiates   this technical component   of Theorem \ref{theo:Homomorphic} (that we stated in the Introduction.)

Note again that homomorphic cryptographic functions, such as
discrete log and RSA encryption,
satisfy an additional property:
They are one-way functions, i,e., tractable to compute but are assumed to be intractable to invert.
Theorem \ref{theo:Homomorphic} (and hence Theorem \ref{IPH}) is inspired by
the concept of partially homomorphic encryption.
However, its focus is not on a one-way encoding of targeted
 nimber-values with impartial games in $\mathbb{I}^P$.
Rather, it characterizes
  the complexity-theoretical
  homomorphism in this classical and natural encoding for impartial games.
Because of the one-way property, RSA and discrete-log functions are decodable by an \cclass{NP}-oracle.
In contrast, the nimber-decoding of impartial games in $\mathbb{I}^P$ is in general \cclass{PSPACE}-hard.

\subsection{Natural Prime Games}
\label{Sec:PrimeGames}

%\begin{definition}[Cartesian  Tree Product]
%The tree sum $G_H = (V_H, E_H)$ of two game trees $G_F = (V_F, E_F)$ and $G_G = (V_G, E_G)$ is a tree with $V_H = V_F \times V_G$ and $E_H$ to be as follows.
%For each for each vertex $v_H \in V_H$ created by $v_F \times v_G$,
%let each child of $v_F$ be set $A$ and each child of $v_G$ be set $B$.
%Then, for each $v_H$,
%add edges ($v_H$, $v_F \times B$) and ($v_H$, $v_G \times A$).
%\end{definition}

%The Tree Sum is a DAG, and not a tree,

Inspired by Conway's notation with \textit{parts} within the context of misere games \cite{ONAG:2001}\cite{siegel2021structure}, we use the following terms to identify what game trees can be described as isomorphically the sum of two other games:

% \begin{definition}[Prime Games]
% An impartial game is a {\em prime} game if its game tree can not be expressed as the product of two trees; otherwise, it is called a {\em composite} game.
% \end{definition}

\begin{definition}[Prime Games and Composite Games]
A game $G$ is a composite game if it is a
%tree
sum of two games that both have tree-height at least 1. Otherwise, it is prime.
\end{definition}

Note that prime games, in a similar manner to prime numbers, can only be summed by a game with tree-height of 0 (ie: just a single vertex) and itself.
It follows from the basic property of Cartesian graph products that each composite game
has a unique decomposition into
prime games.

%TODO: Do I need to prove this, or can I just single the intuition? Below is a "proof" but really its kind of bad
\begin{proposition}[Decomposition in Prime Games]
A game $G$ is isomorphic to a disjunctive sum of two games $A$ and $B$ if and only if
its game tree $T_G$ is
isomorphic to $T_A \blacksquare T_B$.
%the it is the isomophic to the game tree of playing on the tree sum of $a$ and $b$.
\end{proposition}
%\begin{proof}
%First, let's assume that $G$ is a tree sum. Let's label each of the vertices in the game trees of $A$ and $B$ with some representation of the game state. Then, in $G$, we will label each vertex with both of the labels from the cross product. These pairs of states are exactly every possible pair of states for the sum of $A$ and $B$. Note that the possible moves from a state in the sum is to keeping the same position for $A$ and going to an option in $B$ or keeping the same position of $B$ and moving $A$ to an option. There is a move with the correct label for both of these situations. Thus, the tree for the sum is exactly $G$.
%\end{proof}

%Note that two games is sufficient for the sum of any number of games, as $A+B$ itself is a game, so one may sum $A + B$ with $C$ to get $A+B+C$.

%First note that because $G$ in Theorem \ref{IPH} is a prime game, its game tree cannot simply be $G_1 \blacksquare G_2$ (unless one is of height 0).

In Bouton theory for \ruleset{Nim},
 for any non-negative integers $a, b$,
 even though $\nimber(\ruleset{Nim}(a\oplus b)) = \nimber(\ruleset{Nim}(a)) \oplus \nimber(\ruleset{Nim}(b))$,
  $\ruleset{Nim}(a\oplus b)$ is not isomorphic
  to $\ruleset{Nim}(a) \oplus \ruleset{Nim}(b)$.
In fact, $\ruleset{Nim}(a\oplus b)$ is a natural prime game.
Similarly, even though in the proof of Theorem \ref{IPH}, $\rho_{(G_1+G_2)}$ simply copies the syntactic game transition
function that can generate $G_1 \blacksquare G_2$,
  the construction in Theorem \ref{theo:GGComplete}
generates the homomorphic prime game encoding of $\nimber(G_1)\oplus\nimber(G_2)$.

\begin{theorem}[Prime Geography]
\label{theo:Prime}
Each \ruleset{Generalized Geography} position as created in the reduction in Theorem \ref{Theo:Restate1.2} is a prime game.
\end{theorem}
\begin{proof}
Suppose that our game tree is claimed to be $X \blacksquare Y$, with root vertices $x_0$ and $y_0$, respectively, and both $X$ and $Y$ have height at least 1. We will find a contradiction.

%First, note that if the game we reduce from has height 0, the \ruleset{Geography} position starts with only a single move, and thus can't be the sum of two games that both have moves.

Consider vertex $b_0$, an option of $start$.  Since $b_0$ has only one option, that means that it must correspond to a terminal move in either $X$ or $Y$.  WLOG, let it correspond to $x_1$, a terminal vertex in $X$.  Thus, $b_0 = (x_1, y_0)$ and is isomorphic to $Y$, because $x_1$ is terminal in $X$.

The move to $c_1$ must be available, since otherwise the $start$ would have only one option and thus not be a tree sum. This position also has only one option from itself. There are two cases: it either corresponds to a terminal vertex in $X$, say $x_2$, or a terminal vertex in $Y$, say $y_1$.  In the first case, then $c_1 = (x_2, y_0)$, which is isomorphic to $Y$.  This causes a contradiction, however, because the subtrees generated by $b_0$ and $c_1$ are not isomorphic.  ($b_0$ has 2 moves to reach $s_0$, but $c_1$ has 3 moves to reach $s_1$.)

In the second case, $c_1 = (x_0, y_1)$, and $y_1$ is terminal in $Y$.  Then that means there must be a move from $c_1$ to a vertex, $v$, corresponding to $(x_1, y_1)$.  Since both $x_1$ and $y_1$ are terminal (in $X$ and $Y$), that means $v$ will be terminal in the tree sum.  However, $c_1$ doesn't have any options to a terminal vertex.  This case cannot happen and, without any other possible cases, no such $X$ and $Y$ exist as factors for our tree.
\end{proof}

\section{Final Remarks and Open Questions}
\label{Sec:Final}
%We are awesome

It is expected that \cclass{PSPACE}-complete games
encode some valuable secrets.
And once revealed,
those secrets can help players in their decision making (e.g., under the guidance of
 Sprague-Grundy Theory).
In this work, through the lens of computational complexity theory, we see that all polynomially-short impartial games neatly encode their nimber-secrets, which can be efficiently transferred into prime \ruleset{Generalized Geography} games.
The game encoding is so neat that the bitwise-xor of any pair  of these nimber-secrets can be homomorphically re-encoded into another prime game in polynomial time, without the need to
find the secrets first.

We are excited to discover this
 natural mathematical-game-based
\cclass{PSPACE}-complete
homomorphic encoding.
{\em Recreational mathematics can be
simultaneously
serious and fun!}

%There are still many questions that we don't have answers to.

The crypto-concept of (partially) homomorphic encryption has inspired us
to identify these basic
complexity-theoretical properties of this fundamental concept in CGT.
It would have been more fulfilling
if we could also make our findings useful in cryptography.
Currently, we are exploring potential cryptographic applications of this ``game encoding of strategic secrets,'' particularly
on {\em one-way} game generation
for targeted nimbers.
In addition to finding direct cryptographic connections, we are still exploring several concrete CGT questions.
Below, we share some of them.

%\section{Expressiveness of Intractable Games}

\subsection{Expressiveness of Intractable Games: Sprague-Grundy Completeness}
\label{sub:SGC}

In this paper, we have proved that
the \cclass{PSPACE}-complete
polynomially-short
\ruleset{Generalized Geography} is prime Sprague-Grundy complete
for ${\cal I}^P$.
We observe that not all games in ${\cal I}^P$ with \cclass{PSPACE}-hard nimber computation are Sprague-Grundy complete for the family because:

\begin{enumerate}
\item Some intractible games can't encode nimbers polynomially related to the input size
\item Some games with intractible nimber computation have some nim values which are tractible.
\end{enumerate}

For (1), our first example is
\ruleset{Generalized Geography on Degree-Three Graphs}.
In \cite{LichtensteinSipser:1980},
Lichtenstein and Sipser  proved that
\ruleset{Generalized Geography} is \cclass{PSPACE}-complete to solve even when the game graph is planar, bipartite, and has a maximum degree of three. These graph properties are essential to their analysis of
the two-dimensional grid-based \ruleset{Go}.
Mathematically, the maximum achievable nimber in \ruleset{Generalized Geography on Degree-Three Graphs} is three.
Thus, there is no nimber-preserving
  reduction from higher nimber position in $\mathbb{I}^P$ to
  these low-degree \ruleset{Generalized Geography} games.
For the same reason, the \cclass{PSPACE}-complete \ruleset{Atropos} introduced in \cite{DBLP:journals/im/BurkeT08} can not be Sprague-Grundy complete.

\begin{lemma}
    The value of any \ruleset{Atropos} position must be one of these nimbers: $0, \ast, \ast 2, \cdots, \ast 7$.  (And thus, \ruleset{Atropos} cannot be Sprague-Grundy complete.)
    \label{lem:atropos}
\end{lemma}

\begin{proof}
    For the details of how \ruleset{Atropos} is played, please see \cite{DBLP:journals/im/BurkeT08}.  If the last (played) vertex has uncolored neighbors, then there are at most six neighbors, so the highest nimber value is $\ast 6$.

    If the previously-played vertex is fully surrounded by colored vertices\footnote{In this case, the next player gets a ``jump'' and gets to play anywhere on the board.}, then there are two possibilities: either all playable vertices have uncolored neighbors, or some of the playable vertices are also fully surrounded.  In the first case, there may be options to all nimbers $0, \ast, \ast 2, \ldots, \ast 6$, so the value here could be up to $\ast 7$.

    In the second case, the current position, say $G$, is equal to the sum of the portion of the board (say, $H$) without those fully-surrounded (but playable) vertices  and the portion of the board with only those vertices.  Each of those vertices in $G \setminus H$ changes the value by $\ast$. Thus, if there are $k$ of them, $G = H + k \times \ast$. Thus, either $G = H$ or $G = H + \ast$. By the previous case, the nimber of $H$ can be up to $\ast 7$, so the value of $G$ can also be at most $\ast 6$ or $\ast 7$.

    \ruleset{Atropos} has a bounded nimber, so it cannot be Sprague-Grundy complete.
\end{proof}

For (2), both \ruleset{Undirected Geography} \cite{DBLP:journals/tcs/FraenkelSU93} and  \ruleset{Uncooperative Uno} \cite{demaine2014uno} \footnote{In this game, there are two hands, $H_1$ and $H_2$, which each consist of a set of cards. This is a perfect information game, so both players may see each other's hands. Each card has two attributes, a color $c$ and a rank $r$. Each card then thus be represented $(c, r)$. A card can only be played in the center (shared) pile if the previous card matches either the $c$ of the current card or the $r$ of the current card.} are not Sprague-Grundy complete for ${\cal I}^P$---unless \cclass{P} $=$ \cclass{PSPACE}---despite their nimber intractability, with \ruleset{Undirected Geography} being known to have polynomially high nimber positions \cite{BurkeFerlandTengGrundy}.
For \ruleset{Undirected Geography},
Fraenkel, Scheinerman, and Ullman \cite{DBLP:journals/tcs/FraenkelSU93} presented a matching-based characterization to show these games are polynomial-time solvable.
For \ruleset{Uncooperative Uno}, Demaine {\em et al} \cite{demaine2014uno}
presented a polynomial-time reduction to \ruleset{Undirected Geography}.
Thus, any polynomial-time nimber-preserving reduction from ${\cal I}^P$
to \ruleset{Undirected Geography} (or \ruleset{Uncooperative Uno})
would yield a polynomial-time algorithm for solving ${\cal I}^P$.

Rulesets which have nimber preserving reductions from \ruleset{Generalized Geography} are Sprague-Grundy complete.
A simple example is the vertex version of \ruleset{DiGraph Nim} \cite{DBLP:journals/tcs/Fukuyama03},
  in which
  each node has a \ruleset{Nim} pile and players can only
  move to a reachable node in a directed graph from the current node to pick stones.
When every pile has one stone, the game is equivalent to \ruleset{Generalized Geography}
with the underlying graph.
An interesting question is whether \ruleset{Neighboring Nim} (with a polynomial number of stones)---a \cclass{PSPACE}-complete version of \ruleset{Nim} played on an undirected graph \cite{BurkeGeorge}---is Sprague-Grundy complete.

The edge variant of \ruleset{Generalized Geography},
known as the \ruleset{Edge-Geography},
considered in the literature
\cite{DBLP:journals/jcss/Schaefer78,LichtensteinSipser:1980,DBLP:journals/tcs/FraenkelSU93} presents a natural extension.
This is a version of \ruleset{Geography} where instead of deleting the current node after the token moves away,
it is the edge traversed
by the token that is deleted.
\ruleset{Edge-Geography} and
its undirected sub-family, \ruleset{Undirected Edge-Geography} are both \cclass{PSPACE}-complete.
The following proof sketch shows that \ruleset{Edge-Geography} remains Sprague-Grundy complete.

\begin{corollary}
\ruleset{Edge-Geography} is prime Sprague-Grundy-complete for $\cal{I}^P$.
\end{corollary}
\begin{proof}
We can follow the early parts of the proof for \ruleset{Generalized Geography}. We reduce from all polynomially-short games, creating a game of \ruleset{Edge-Geography} for each. Then, for each game, we again append two ``filler'' moves to the beginning, to ensure that it is exactly 0 or $\ast$.

We can then reuse our scheme from figure \ref{fig:Qs}. Since there are no cycles in that gadget, play between both \ruleset{Edge-Geography} and \ruleset{Generalized Geography} is identical.

Of course, the primality section required knowing that the main \ruleset{Geography} game didn't start with an out degree of only one. To fix this, we can simply have $v_b$ go to $v_{a1}$ and $v_{a2}$ which both only have a single edge to $v_{s}$.
\end{proof}

It remains open whether
\ruleset{Undirected Edge-Geography} is Sprague-Grundy complete.

In addition to these
rulesets
adjacent to \ruleset{Generalized Geography}, we are interested in the following three
well-studied games:

\begin{itemize}
\item Is \ruleset{Node Kayles}
Sprague-Grundy complete for ${\cal I}^P$?
Is \ruleset{Avoid True}
Sprague-Grundy complete for ${\cal I}^P$?  Is \ruleset{Generalized Chomp} Sprague-Grundy complete for ${\cal I}^P$?
\end{itemize}
In our proof for \ruleset{Generalized Geography}, we critically use the ``locality'' in this graph-theoretical game:
The options are defined by the graph-neighbors of the current node.
Both \ruleset{Node-Kayles} and \ruleset{Avoid True} are far more ``global''; there is no need for moves to be near the previous move.
We are also interested in \ruleset{Generalized Chomp} because the hierarchical structures from partial orders could be instrumental to analyses.

\ruleset{Node-Kayles}---see below for more discussion---also suggests
the following basic structural question:
\begin{itemize}
\item Is there a natural ruleset in ${\cal I}^P$ that is Sprague-Grundy-complete for ${\cal I}^P$ but not {\bf\em prime}  Sprague-Grundy-complete for ${\cal I}^P$?
\end{itemize}

%We know that \ruleset{Directed Neighboring Nim} is trivially complete for the class, since \ruleset{Geography} is just a version of the game with limitations.

\subsection{Game Encoding and Computational Homomorphism}

Let's call a family ${\cal H}$ of impartial rulesets satisfying
Theorem \ref{theo:Homomorphic} (in place of ${\cal I}^{P}$ and with a prime game of ${\cal I}^{P}$ in place of \ruleset{Generalized Geography})
a {\em computationally-homomorphic family}.
Note that for any ${\cal J}$ including \ruleset{Undirected Geography}, ${\cal J}$ satisfies  Theorem \ref{theo:Homomorphic}.
%We call such $J$ a \cclass{PSPACE}-hard Homomorphic Encoding-Family.

Now suppose
we ``slightly'' weaken
Theorem \ref{theo:Homomorphic}
by removing the prime-game
requirement (in the {\em Homomorphic Game Encoding} condition), and call
$\cal H$ satisfying the weakened version of Theorem \ref{theo:Homomorphic}
a
{\em weakly computationally-homomorphic family}.
Then,
${\cal I}^P$ itself is
a weakly computationally-homomorphic family, by Sprague-Grundy Theory
and the fact that ${\cal I}^P$ is closed under the disjunctive sum.

Indeed,
if  a ruleset in ${\cal I}^P$ is \cclass{PSPACE}-complete and allows a simple way to express the sum of two positions as a single position, then the ruleset is a
weakly computationally-homomorphic family.
One of the most basic examples of this is \ruleset{Node Kayles}. Here, two positions can be trivially summed into a single game by simply taking the two graphs and making them a single (disconnected) graph.

This is a very common property for combinatorial games to have. However, many impartial games with this property aren't known to be intractable. As an example, \ruleset{Cram} is a game that is simply played by placing 2x1 dominoes in either horizontal or vertical orientation on unoccupied tiles of a 2-dimensional grid.  Two \ruleset{Cram} positions can be added together by surrounding each with a boundary of dominoes, then concatenating the two boards together.  Unfortunately, it is not currently known whether \ruleset{Cram} is intractable. % and can clearly be added together by simply having a larger board with all the tiles that are in neither being filled in with dominoes.

%If a family of impartial rulesets ${\cal J}$ satisfying Theorem 1.1 except  the requirement that the homomorphic re-encoding of two games has to be a prime game, then we call ${\cal J}$a weakly \cclass{PSPACE}-hard Homomorphic Encoding-Family.

Related to the question we asked in
Section \ref{sub:SGC}, we are curious to know:

\begin{itemize}
\item Given a pair of \ruleset{Node Kayles}
positions $G_1$ and $G_2$, can we construct, in polynomial time, a {\bf\em prime}
\ruleset{Node Kayles}
position satisfying
$\nimber(G) = \nimber(G_1)\oplus \nimber(G_2)$?
\end{itemize}

\subsection{Beyond ${\cal I}^P$}
More generally,

\begin{itemize}
\item Are there analog extensions of our results to polynomially-short partizan games?
\item
Is there a characterization of
Sprague-Grundy completeness for
${\cal I}^P$?

\item
Does the family of \cclass{PSPACE}-solvable impartial games have a natural Sprague-Grundy-complete ruleset?

\item
Does the family of all impartial games have a natural Sprague-Grundy-complete ruleset?

\item What is the complexity of \ruleset{Graph Nim} with an exponential number of stones?
\end{itemize}

For these last few questions, we may need to go beyond \cclass{PSPACE} as well as polynomially-short games to unlock the nimber secrets.

\subsection{Finally}

Is there a Bouton analog---i.e., a more clean and direct graph operator---to compute a
\ruleset{Generalized Geography} game $G$ from two \ruleset{Generalized Geography} games $G_1$ and $G_2$
such that $\nimber(G) = \nimber(G_1)\oplus\nimber(G_2)$?

\bibliographystyle{plain}

\begin{thebibliography}{10}

\bibitem{akl2010importance}
Selim~G Akl.
\newblock On the importance of being quantum.
\newblock {\em Parallel processing letters}, 20(03):275--286, 2010.

\bibitem{applegatejacobsonsleator:1991}
D.~Applegate and D.~Sleator G.~Jacobson.
\newblock Computer analysis of sprouts.
\newblock Technical Report CMU-CS-91-144, Carnegie Mellon University Computer
  Science, 1991.

\bibitem{DBLP:journals/tcs/BeaulieuBD13}
Gabriel Beaulieu, Kyle~G. Burke, and {\'E}ric Duch{\^e}ne.
\newblock Impartial coloring games.
\newblock {\em Theoret. Comput. Sci.}, 485:49--60, 2013.

\bibitem{WinningWays:2001}
Elwyn~R. Berlekamp, John~H. Conway, and Richard~K. Guy.
\newblock {\em Winning Ways for your Mathematical Plays}, volume~1.
\newblock A K Peters, Wellesley, Massachsetts, 2001.

\bibitem{Bouton:1901}
Charles~L. Bouton.
\newblock Nim, a game with a complete mathematical theory.
\newblock {\em Annals of Mathematics}, 3(1/4):pp. 35--39, 1901.

\bibitem{BurkeFerlandTengQCGT}
Kyle Burke, Matthew Ferland, and Shang{-}Hua Teng.
\newblock Quantum combinatorial games: Structures and computational complexity.
\newblock {\em CoRR}, abs/2011.03704, 2020.

\bibitem{burke2021transverse}
Kyle Burke, Matthew Ferland, and Shang{-}Hua Teng.
\newblock Transverse wave: an impartial color-propagation game inspired by
  social influence and quantum nim.
\newblock {\em CoRR}, abs/2101.07237, 2021.

\bibitem{BurkeFerlandTengGrundy}
Kyle Burke, Matthew Ferland, and Shang-Hua Teng.
\newblock Winning the war by (strategically) losing battles: Settling the
  complexity of grundy-values in undirected geography.
\newblock In {\em Proceedings of the 62nd Annual Symposium on Foundations of
  Computer Science (FOCS)}. IEEE, 2021.

\bibitem{BurkeGeorge}
Kyle~W. Burke and Olivia George.
\newblock A pspace-complete graph nim.
\newblock {\em CoRR}, abs/1101.1507, 2011.

\bibitem{DBLP:journals/im/BurkeT08}
Kyle~W. Burke and Shang-Hua Teng.
\newblock Atropos: A pspace-complete sperner triangle game.
\newblock {\em Internet Mathematics}, 5(4):477--492, 2008.

\bibitem{ONAG:2001}
John~H. Conway.
\newblock {\em On numbers and games {(2.} ed.)}.
\newblock A {K} Peters, 2001.

\bibitem{demaine2014uno}
Erik~D Demaine, Martin~L Demaine, Nicholas~JA Harvey, Ryuhei Uehara, Takeaki
  Uno, and Yushi Uno.
\newblock Uno is hard, even for a single player.
\newblock {\em Theoretical Computer Science}, 521:51--61, 2014.

\bibitem{dorbec2017toward}
Paul Dorbec and Mehdi Mhalla.
\newblock Toward quantum combinatorial games.
\newblock {\em arXiv preprint arXiv:1701.02193}, 2017.

\bibitem{EvenTarjanHex}
S.~Even and R.~E. Tarjan.
\newblock A combinatorial problem which is complete in polynomial space.
\newblock {\em J. ACM}, 23(4):710–719, October 1976.

\bibitem{fraenkel2004complexity}
Aviezri~S Fraenkel.
\newblock Complexity, appeal and challenges of combinatorial games.
\newblock {\em Theoretical Computer Science}, 313(3):393--415, 2004.

\bibitem{DBLP:journals/jct/FraenkelL81}
Aviezri~S. Fraenkel and David Lichtenstein.
\newblock Computing a perfect strategy for n x n chess requires time
  exponential in n.
\newblock {\em J. Comb. Theory, Ser. {A}}, 31(2):199--214, 1981.

\bibitem{DBLP:journals/tcs/FraenkelSU93}
Aviezri~S. Fraenkel, Edward~R. Scheinerman, and Daniel Ullman.
\newblock Undirected edge geography.
\newblock {\em Theor. Comput. Sci.}, 112(2):371--381, 1993.

\bibitem{DBLP:journals/tcs/Fukuyama03}
Masahiko Fukuyama.
\newblock A nim game played on graphs.
\newblock {\em Theor. Comput. Sci.}, 1-3(304):387--399, 2003.

\bibitem{Gale:1979}
David Gale.
\newblock The game of {H}ex and the {B}rouwer fixed-point theorem.
\newblock {\em American Mathematical Monthly}, 10:818--827, 1979.

\bibitem{glos2019role}
Adam Glos and Jaros{\l}aw~Adam Miszczak.
\newblock The role of quantum correlations in cop and robber game.
\newblock {\em Quantum Studies: Mathematics and Foundations}, 6(1):15--26,
  2019.

\bibitem{goff2006quantum}
Allan Goff.
\newblock Quantum tic-tac-toe: A teaching metaphor for superposition in quantum
  mechanics.
\newblock {\em American Journal of Physics}, 74(11):962--973, 2006.

\bibitem{Grier}
Daniel Grier.
\newblock Deciding the winner of an arbitrary finite poset game is
  pspace-complete.
\newblock In {\em Proceedings of the 40th International Conference on Automata,
  Languages, and Programming - Volume Part I}, ICALP'13, page 497–503,
  Berlin, Heidelberg, 2013. Springer-Verlag.

\bibitem{Grundy:1939}
P.~M. Grundy.
\newblock Mathematics and games.
\newblock {\em Eureka}, 2:198---211, 1939.

\bibitem{LichtensteinSipser:1980}
David Lichtenstein and Michael Sipser.
\newblock Go is polynomial-space hard.
\newblock {\em J. ACM}, 27(2):393--401, 1980.

\bibitem{NashHex}
John~F. Nash.
\newblock {\em Some Games and Machines for Playing Them}.
\newblock RAND Corporation, Santa Monica, CA, 1952.

\bibitem{PapadimitriouBook:1994}
C.~H. Papadimitriou.
\newblock {\em Computational Complexity}.
\newblock Addison Wesley, Reading, Massachsetts, 1994.

\bibitem{Reisch:1981}
S.~Reisch.
\newblock Hex ist {PSPACE}-vollst{\"a}ndig.
\newblock {\em Acta Inf.}, 15:167--191, 1981.

\bibitem{DBLP:journals/jcss/Schaefer78}
Thomas~J. Schaefer.
\newblock On the complexity of some two-person perfect-information games.
\newblock {\em Journal of Computer and System Sciences}, 16(2):185--225, 1978.

\bibitem{siegel2021structure}
Aaron~N. Siegel.
\newblock On the structure of mis\`ere impartial games, 2021.

\bibitem{PosetGame}
Michael Soltys and Craig Wilson.
\newblock On the complexity of computing winning strategies for finite poset
  games.
\newblock {\em Theory Comput. Syst.}, 48:680--692, 04 2011.

\bibitem{Sprague:1936}
R.~P. Sprague.
\newblock \"{U}ber mathematische {K}ampfspiele.
\newblock {\em T\^{o}hoku Mathematical Journal}, 41:438---444, 1935-36.

\bibitem{wolfe2000go}
David Wolfe.
\newblock Go endgames are pspace-hard.
\newblock {\em intelligence}, 9(7):6, 2000.

\bibitem{Zeilberger:2004}
D.~Zeilberger.
\newblock Chomp, recurrences and chaos.
\newblock {\em Journal of Difference Equations and Applications}, 10:1281 --
  1293, 2004.

\end{thebibliography}

\end{document}